\documentclass[11pt]{llncs}
\usepackage[letterpaper,hmargin=1in,vmargin=1in]{geometry}

\usepackage{makeidx}
\usepackage{amsfonts}
\usepackage{amsmath}
\usepackage{graphicx}
\usepackage{algorithm}
\usepackage{algorithmic}
\usepackage[USenglish]{babel}

\usepackage{amssymb}
\begin{document}

\pagestyle{headings}
%In order to omit page numbers and running heads
%please change this line to
%\pagestyle{empty}
%and change the first command line too, see above.

%\documentclass{article}

%\usepackage{graphics}
%\usepackage[section]{placeins}

%\newtheorem{theorem}{Theorem}[section]
%\newtheorem{lemma}[theorem]{Lemma}
%\newtheorem{proposition}[theorem]{Proposition}
%\newtheorem{corollary}[theorem]{Corollary}
\newtheorem{observation}{Observation}

\title{Worst-case optimal approximation algorithms for maximizing triplet
consistency within phylogenetic networks\thanks{This research has been partly funded
by the Dutch BSIK/BRICKS project, and by the EU Marie Curie Research Training Network 
ADONET, Contract No MRTN-CT-2003-504438.}}

\author{Jaroslaw Byrka\inst{1,2},
	Pawel Gawrychowski\inst{4},
	Katharina T. Huber\inst{3},
	Steven Kelk\inst{1},
}

\institute{
Centrum voor Wiskunde en Informatica,\\
Kruislaan 413, NL-1098 SJ Amsterdam, Netherlands\\
\email{\{J.Byrka,S.M.Kelk\}@cwi.nl}\\
\and
Eindhoven University of Technology,\\ 
P.O. Box 513, 5600 MB Eindhoven, The Netherlands\\
\and
School of Computing Sciences, 
University of East Anglia,\\
Norwich, NR4 7TJ, United Kingdom\\
\email{katharina.huber@cmp.uea.ac.uk}
\and
Institute of Computer Science,
University of Wroclaw,\\
ul. Joliot-Curie 15, 50-383 Wroclaw, Poland
}

\maketitle

\begin{abstract}
This article concerns the following question arising in computational
evolutionary biology. For a given subclass of phylogenetic networks, what
is the maximum value of
$0 \leq p \leq 1$ such that for every input set $T$ of rooted triplets, there
exists some network ${\mathcal N}(T)$ from the subclass such that at least $p|T|$ of the
triplets are consistent with ${\mathcal N}(T)$? Here we prove that the set containing
all triplets (the full triplet set)
in some sense defines $p$, and moreover that any network $\mathcal{N}$ achieving fraction $p'$ for the full triplet
set can be converted in polynomial time into an isomorphic network $\mathcal{N}'(T)$ achieving fraction $\geq p'$ for
an arbitrary triplet set $T$. We demonstrate the power of this result for the field of phylogenetics by giving worst-case optimal
algorithms for level-1 phylogenetic networks (a much-studied extension of phylogenetic trees),
improving considerably upon the $5/12$ fraction obtained recently by Jansson, Nguyen and
Sung in \cite{JS2}. For level-2 phylogenetic networks we show that $p \geq 0.61$. 
We note that all the results in this article also apply to weighted triplet sets.

\end{abstract}

%\pagebreak

\section{Introduction}

One of the most commonly encountered problems in computational evolutionary biology
is to plausibly infer the evolutionary history of a set of species, often
abstractly modelled as a tree, using obtained biological data.
Existing algorithms for directly constructing such a tree do not scale
well (in terms of running time) and this has given rise to \emph{supertree} methods:
first infer trees for small subsets of the species and then puzzle them
together into a bigger tree such that in some well-defined sense the information in
the subset trees is preserved \cite{B04}. In the fundamental case where the subsets in question each
contain exactly three species - subsets of two or fewer species cannot convey
information - we speak of \emph{rooted triplet} methods.\\
\\
In recent years improved understanding of the complex mechanisms driving evolution has
stimulated interest in reconstructing evolutionary \emph{networks} \cite{dolittle}\cite{kunin}\cite{martin}\cite{rivera}. Such 
structures are more general than trees and allow us to capture the phenomenon of reticulate
evolution i.e. non tree-like evolution caused by processes such as hybrid speciation and horizontal gene transfer. A natural abstraction of reticulate evolution,
used already in several papers, is to permit \emph{recombination vertices}, vertices with indegree greater than one.
Informally a \emph{level-$k$ phylogenetic network} is an evolutionary network in which each biconnected
component contains at most $k$ such recombination vertices. Phylogenetic trees form the base: they are
level-0 networks. The higher the level of a network, the more intricate the pattern
of reticulate evolution that it can accommodate. Note that phylogenetic networks can also be useful for visualising two
or more competing hypotheses about tree-like evolution.\\
\\
Various authors have already studied the problem of constructing phylogenetic trees
(and more generally networks) which are consistent with an input set of rooted triplets.
Aho et al \cite{aho} showed a simple polynomial-time algorithm which, given a set
of rooted triplets, finds a phylogenetic tree consistent with all the triplets, or
shows that no such tree exists. For the equivalent problem in level-1 and level-2 networks
the problem becomes NP-hard \cite{iersel}\cite{JS2}, although the problem becomes polynomial-time solveable if
the input triplets are \emph{dense} i.e. if there is at least one triplet in the input
for each subset of three species \cite{iersel}\cite{JS1}.\\
\\
Several authors have considered algorithmic strategies of use when the algorithms from \cite{aho} and
\cite{JS1} fail to find a tree or network. G\c{a}sieniec et al \cite{Gasieniec99} gave a polynomial-time
algorithm which always finds a tree consistent with at least 1/3 of the (weighted) input triplets,
and furthermore showed that 1/3 is best possible when all possible triplets on $n$ species (the \emph{full triplet set})
are given as input. On the negative side, \cite{bryant97}\cite{jansson01}\cite{wu04} showed that it is NP-hard to
find a tree consistent with a maximum number of input triplets.
In the context of level-1 networks, \cite{JS2} gave a polynomial-time algorithm which produces a level-1
network consistent with at least $5/12 \approx 
0.4166$ of the input triplets. They also described an upper-bound, which is a function of the number of distinct species 
$n$ in the input,
on the percentage of input triplets that can be consistent with a level-1 network. As in \cite{Gasieniec99} this upper bound 
is tight in the sense that it is 
the best possible for the full triplet set on $n$ species. They computed a value of $n$ for which their upper bound equals
approximately $0.4883$, showing that in general a fraction better than this is not possible. The apparent convergence
of this bound from above to $0.4880...$ begs the question, however, whether a 
fraction better than $5/12$ is possible for level-1 networks, and whether the full 
triplet set is in general always the worst-case scenario for such fractions.\\
\\
In this paper we answer these questions in the affirmative, and in fact we give a much stronger
result. In particular, we develop a
probabilistic argument that (as far as such fractions are concerned) the full triplet set is indeed always the worst possible case, 
irrespective of the type of network being studied (Proposition~\ref{prop:random}, Corollary~\ref{lem:randomcor}). Furthermore, by using
a generic, derandomized polynomial-time (re)labelling procedure we can convert a network $\mathcal{N}$ which achieves a 
fraction $p'$ for the full triplet set into an isomorphic network $\mathcal{N}'(T)$ that achieves a fraction $\geq p'$ for a given 
input triplet
set $T$ (Theorem~\ref{thm:genericderand}). In this way we can easily use the full triplet set to generate, for any network structure,
a lower bound on the fraction that can be achieved for arbitrary triplet sets within such a network
structure. The derandomization we give is fully general (with a highly optimized running time) and leads immediately to a simple extension
of the 1/3 result from \cite{Gasieniec99}. For level-1 networks we use the derandomization to give a 
polynomial-time algorithm which
achieves a fraction \emph{exactly equal} to the level-1 upper-bound identified in \cite{JS2},
and which is thus worst-case optimal for level-1 networks. 
We formally prove that this achieves a fraction of at least 
0.48 for all $n$. Finally, we demonstrate the flexibility of our technique by proving that for level-2 networks (see \cite{iersel}) we can, 
for any triplet set $T$, find in polynomial time a level-2 network consistent with at least a fraction 0.61 of the triplets in $T$ (Theorem 
\ref{thm:lev2}).\\
\\
We emphasize that in this article we are optimizing (and thus
defining worst-case optimality) with respect to $|T|$, the number of triplets in the input, not $Opt(T)$, the
size of the optimal solution for that specific $T$. The latter formulation we call the MAX variant of the problem.
The fact that $Opt(T)$ is always bounded above by $|T|$ implies that an algorithm that obtains a fraction $p'$ of the input $T$ is
trivially also a $p'$-approximation for the corresponding MAX problem. Better approximation
factors for the MAX problem might, however, be possible. We discuss this further in Section \ref{sec:conc}.\\
\\
The results in this article are given in terms of unweighted triplet sets. A natural extension,
especially in phylogenetics, is to attach a weight to each triplet $t \in T$ i.e. a value $w(t) \in \mathbb{Q}_{\geq 0}$
denoting the relative importance of (or confidence in) $t$. In this weighted version of the problem
fractions are defined relative to the total weight of $T$ (defined as the sum of the weights of all triplets in $T$), not
to $|T|$. It is easy to verify that all the results in this article also hold for the weighted version of the problem.

\section{Definitions}
\label{sec:def}

A \emph{phylogenetic network} (\emph{network} for short) $\mathcal{N}$ on species
set $X$ is defined as a pair $(N, \gamma)$ where $N$ is the \emph{network topology}
(\emph{topology} for short) and $\gamma$ is a \emph{labelling} of the topology. The topology
is a directed acyclic graph in which exactly one vertex has indegree 0 
and outdegree 2 (the root) and all other vertices
have either indegree 1 and outdegree 2 (\emph{split vertices}), indegree 2 and outdegree 1 (\emph{recombination vertices}) or
indegree 1 and outdegree 0 (\emph{leaves}). A labelling is a bijective mapping from
the leaf set of $N$ (denoted $L^N$) to $X$. Let $n = |X| = |L^N|$. \\
\\
A directed acyclic graph is \emph{connected} (also called ``weakly connected'')
if there is an undirected path between any two vertices and \emph{biconnected} if it contains no vertex whose removal
disconnects the graph. A \emph{biconnected component} of a network is a maximal biconnected subgraph.

\begin{definition}
A network is said to be a \emph{level-$k$ network} if each biconnected component contains at most $k \in \mathbb{N}$ recombination
vertices.
\end{definition}
We define \emph{phylogenetic trees} to be the class of level-0 networks.

\begin{figure}
\centering \vspace{-0.5cm} \includegraphics{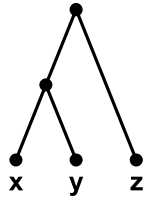} \caption{One of the three possible triplets on the set
of species $\{x, y, z\}$. Note that, as with all figures in this article, all arcs are assumed to be directed downwards,
away from the root.} \vspace{-0.3cm} \label{fig:singletriplet}
\end{figure}

The unique \emph{rooted triplet} (\emph{triplet} for short) on a species set $\{x,y,z\} \subseteq X$ in which the lowest
common ancestor of $x$ and $y$ is a proper descendant of the lowest common ancestor of $x$ and $z$
is denoted by $xy|z$ (which is identical to $yx|z$). For any set $T$ of triplets define $X(T)$ as the union of the
species sets of all triplets in $T$. 

\begin{definition}
\label{def:con} A triplet $xy|z$ is \emph{consistent} with a network $\mathcal{N}$ (interchangeably: $\mathcal{N}$ is consistent 
with $xy|z$) if $\mathcal{N}$ contains a subdivision of $xy|z$, i.e. if $\mathcal{N}$ contains vertices $u \neq v$ and pairwise 
internally vertex-disjoint paths $u \rightarrow x$, $u \rightarrow y$, $v \rightarrow u$ and $v \rightarrow z$\footnote{
Where it is clear from the context, as in this case, we may refer to a leaf by the species that it is mapped to.}.
\end{definition}
By extension, a set of triplets $T$ is consistent with a network $\mathcal{N}$ (interchangeably: $\mathcal{N}$ is consistent with 
$T$) iff, for all $t \in T$, $t$ is consistent with $\mathcal{N}$.
Checking consistency can be done in polynomial time, see Lemmas~\ref{lem:qubic_consistency} and~\ref{lem:fast_consistency}.
%\begin{lemma}
%\label{lem:con} Given a phylogenetic network $\mathcal{N}$ with $n_{+} = |V(\mathcal{N})|$ vertices, and a triplet $t$, it is 
%possible to determine in time $O( n_{+}^4 )$ whether $t$ is consistent with $\mathcal{N}$. THIS HAS TO CHANGE! NEED TO DECIDE
%ABOUT WHAT TO DO WITH $n_{+}$ BEFORE WE PUT A RUNNING TIME HERE.
%\end{lemma}
%\begin{proof}
%See Appendix.
%\qed
%\end{proof}

\begin{figure}
\centering \vspace{-0.5cm} \includegraphics{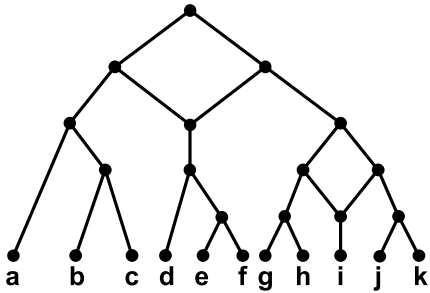} \caption{An example of a level-1 network.} \vspace{-0.3cm} \label{fig:bignetwork}
\end{figure}

\section{Labelling a network topology}
\label{sec:label}
Suppose we are given a topology $N$ with $n$ leaves, and a set $T$ of triplets where
$L^N = \{ l_1, l_2, \ldots, l_n\}$ and $X = X(T) = \{ x_1, x_2, \ldots, x_n\}$.
The specific goal of this section is to create a labelling $\gamma$
such that the number of triplets from $T$ consistent with $(N,\gamma)$, is
maximized.

Let $f(N,\gamma,T)$ denote the fraction of $T$ that is consistent with $(N, \gamma)$.

Consider the special set $T_1(n)$, the \emph{full triplet set}, of all the possible $3\binom{n}{3}$ triplets with leaves
labelled from $\{x_1, x_2, \ldots, x_n \}$. Observe that for this triplet set the number of triplets consistent with a phylogenetic
network $(N, \gamma)$ does not depend on the labelling $\gamma$. We may %thus unambiguously
thus define $\#N = f(N, \gamma, T_1(n)) = f(N, T_1(n))$ by considering any arbitrary, fixed labelling $\gamma$. (Note that
$\#N = 1/3$ if $N$ is a tree topology.)

We will argue that the triplet set $T_1(n)$ is the worst-case input for maximizing
$f(N,\gamma,T)$ for any fixed topology $N$ on $n$ leaves. In particular we prove 
the following:

\begin{proposition}
\label{prop:random}
For any topology $N$ with $n$ leaves and any set of
triplets $T$, if the labelling $\gamma$
is chosen uniformly at random, then the quantity $f(N,\gamma,T)$ is a random variable
with expected value $E(f(N,\gamma,T))=\#N$.
\end{proposition}

\begin{proof}
Consider first the full triplet set $T_1(n) = \{t_1, t_2, \ldots, t_{3\binom{n}{3}}\}$ 
and an arbitrary fixed labelling $\gamma_0$.
By labelling $N$ we fix the \emph{position}
of each of the triplets in $N$. Formally, a position
of a triplet $t = xy|z$ (with respect to $\gamma_0$) is a triplet
$p = \gamma_0^{-1}(t) = \gamma_0^{-1}(x)\gamma_0^{-1}(y)|\gamma_0^{-1}(z)$ on the leaves of $N$.
We may list possible positions for a triplet in $N$
as those corresponding to $t_1, t_2, \ldots, t_{3\binom{n}{3}}$ in
$(N, \gamma_0)$. Since a $\#N$ fraction of $t_1, t_2, \ldots, t_{3\binom{n}{3}}$
is consistent with $(N, \gamma_0)$, a $\#N$ fraction of these positions
makes the triplet consistent. Now consider a single triplet $t \in T$ and a labelling $\gamma$
that is chosen randomly from the set $\Gamma$ of $n!$ possible bijections
from $L^N$ to $X$. Observe, that for each $t_i \in T_1(n)$, exactly
$2 \cdot (n-3)!$ labellings $\gamma \in \Gamma$ make triplet $t$ have
the same position in $(N, \gamma)$ as $t_i$ has in $(N, \gamma_0)$
(the factor of $2$ comes from the fact that we think of $xy|z$ and 
$yx|z$ as being the same triplet). Any single labelling occurs with 
probability $\frac{1}{n!}$, hence triplet $t$ takes any single position
with probability $\frac{2\cdot(n-3)!}{n!} = \frac{1}{3\cdot{n\choose3}}$.

Since for an arbitrary $t \in T$ each of the $3\cdot{n\choose3}$ positions
have the same probability and $\#N$ of them make $t$ consistent,
the probability of $t$ being consistent with $(N, \gamma)$ is $\#N$.
The expectation is thus that a fraction $\#N$ of the triplets in $T$
are consistent with $(N,\gamma)$.
\qed
\end{proof}

From the expected value of a random variable we may conclude the existence
of a realization that attains at least this value.

\begin{corollary}
\label{lem:randomcor}
For any topology $N$ and any set of triplets $T$ there exists a labelling $\gamma_0$
such that $f(N,\gamma_0,T) \geq \#N$.
\end{corollary}

We may deterministically find such a $\gamma_0$ by derandomizing 
the argument in a greedy fashion, using the method of conditional
expectation (see e.g. \cite{randomized}). In particular, we will show the following.
\begin{theorem}
\label{thm:genericderand}
For any topology $N$ and any triplet set $T$, a labelling $\gamma_0$ such
that $f(N,\gamma_0,T) \geq \#N$ can be found in time $O(m^3 + n|T|)$, where $m$ and $n$
are the numbers of vertices and leaves of $N$.
\end{theorem}

We use standard arguments from conditional expectation to prove that the
solution is of the desired quality. It is somewhat more sophisticated to
optimize the running time of this derandomization procedure. We therefore dedicate the following
subsection to the proof of Theorem~\ref{thm:genericderand}.

%\documentclass[10pt,a4paper,oneside]{article}

%\usepackage{amsmath,amssymb,amsthm,algorithm}
%\usepackage[T1]{fontenc}
%\usepackage{tikz}
%\usetikzlibrary{automata}
%\usetikzlibrary{snakes}

%\newtheorem{definition}{Definition}
%\newtheorem{theorem}{Theorem}
%\newtheorem{lemma}{Lemma}
%\newtheorem{fact}{Fact}

%\begin{document}

\subsection{An optimized derandomization procedure}

We start with a sketch of the derandomization procedure.
\begin{enumerate}
% \item Preprocessing: compute which positions in $\mathcal{N}$ make a triplet consistent.
 \item $\Gamma \leftarrow$ set of all possible labellings.
 \item while there is a leaf $l$ whose label is not yet fixed, do:
	\begin{itemize}
	 \item for every species $x$ that is not yet used
		\begin{itemize}
		 \item let $\Gamma_x$ be the set of labellings from $\Gamma$ where $l$ is labeled by $x$
		 \item compute $E(f(N,\gamma_x,T))$, where $\gamma_x$ is chosen randomly
			from $\Gamma_x$
		\end{itemize}
	 \item $\Gamma \leftarrow \Gamma_x$ 
  s.t. $x = argmax_{x}{E(f(N,\gamma_x,T))}$
	\end{itemize}
 \item return $\gamma_0 \leftarrow$ the only element of $\Gamma$.
\end{enumerate}

The following lemma shows that the solution produced is of the desired quality.
\begin{lemma} \label{lem:quality}
 $f(N,\gamma_0,T) \geq \#N$.
\end{lemma}
\begin{proof}
By Proposition~\ref{prop:random} the initial random labelling $\gamma$
has the property $E(f(N,\gamma,T))=\#N$.
It remains to show that this expectation is not decreasing when labels of leaves 
get fixed during the algorithm. Consider a single update $\Gamma \leftarrow \Gamma_x$
of the range of the random labelling. By the choice of the leaf $l$ to get a fixed label we choose 
a partition of $\Gamma$ into blocks $\Gamma_x$.
The expectation $E(f(N,\gamma,T))$ is an average of $f(N,\gamma,T)$ over $\Gamma$,
and at least one of the blocks $\Gamma_x$ of the partition has this average at least as big as the total average.
Hence, by the choice of $\Gamma_x$ with the highest expectation of $f(N,\gamma_x,T)$, 
we get $E(f(N,\gamma_x,T)) \geq E(f(N,\gamma,T))$.
\qed
\end{proof}

We will now propose an efficient implementation of the derandomization procedure.
Since the procedure takes a topology $N$ as an input,
we need to express the running time of this procedure in terms of the size of $N$.
We will use $m = |V(N)|$ for the number of vertices of $N$, and $n$ for the number
of leaves.

We will need a generalized definition of consistency.
In Definition~\ref{def:con} it is assumed that $x,y$ and $z$ are leaves
of the network $\mathcal{N}$. In the following we will use a predicate $consistent(x,y,z)$
defined on vertices of $\mathcal{N}$ which coincides with Definition~\ref{def:con} when 
$x,y$ and $z$ are leaves. The predicate $consistent(x,y,z)$ is defined to be true
if $\mathcal{N}$ contains vertices $u \neq v$ and pairwise 
internally vertex-disjoint paths $u \rightarrow x$, $u \rightarrow y$, $v \rightarrow u$ and $v \rightarrow z$;
paths are allowed to be of length 0 (e.g. $u = x$), but vertices $x,y$ and $z$ need to be all different.
We begin with a consistency checking algorithm.

\begin{lemma}\label{lem:qubic_consistency}
Given a network $\mathcal{N}$ with $m$ vertices, we can preprocess it in time $O(m^3)$ so that given any three vertices $x,y,z$ we can check whether
$xy|z$ is consistent with $\mathcal{N}$ in time $O(1)$.
\end{lemma}

\begin{proof}
Assume that vertices of $\mathcal{N}$ are numbered according to a topological order so that whenever $(u,v)$ is an edge, $u<v$.
Observe that $consistent(x,y,z)$ can be expressed in terms of $consistent(x',y',z')$ where $\max(x',y',z')<\max(x,y,z)$
as follows:

\begin{enumerate}
\item $x,y<z$. Then $consistent(x,y,z)$ holds iff for some $z'\neq x,y$ $(z',z)$ is an edge and $consistent(x,y,z')$ holds.
\item $x,z<y$. Then $consistent(x,y,z)$ holds iff $(x,y)$ is an edge and $join(x,z)$ holds or for some $y'\neq x,z$ $(y',y)$ is an edge and
$consistent(x,y',z)$ holds.
\item $y,z<x$. Then $consistent(x,y,z)$ holds iff $(y,x)$ is an edge and $join(y,z)$ holds or for some $x'\neq y,z$ $(x',x)$ is an edge and
$consistent(x',y,z)$ holds.
\end{enumerate}

where $join(x,z)$ is a predicate stating that $\mathcal{N}$ contains $t\neq x$ and internally vertex-disjoint paths $t\rightarrow x$ and
$t\rightarrow z$. It can be verified in linear time, for example by 
checking if there is a path $r\rightarrow z$ (where $r$ is the root) after 
removing $x$ from the network.

So, determining all consistent triplets can be done by calculating all $consistent(x,y,z)$. The number of operations we have to perform
for each $x,y,z$ is not greater than sum of indegrees of those vertices which makes the overall complexity $O(m^3+m^2|E(\mathcal{N})|)=O(m^3)$.
\qed
\end{proof}

Although we do not use the result in this article, we note that the above lemma can be strengthened to obtain the following, which
is proved in the appendix.

\begin{lemma} \label{lem:fast_consistency}
Given a level-$k$ network $\mathcal{N}$, we can preprocess it in time $O(m+mk^2)$ so that given any three vertices $x,y,z$ we can check whether
$xy|z$ is consistent with $\mathcal{N}$ in time $O(1)$.
\end{lemma}

Armed with Lemma \ref{lem:qubic_consistency} (or Lemma \ref{lem:fast_consistency}) we are now ready to
proceed with the derandomization. Assume that we assign labels to leaves in order of
their position in the topological ordering of Lemma \ref{lem:qubic_consistency}; for simplicity let us refer
to the leaves as $1,2,\ldots,n$. The most time-consuming part of the algorithm is 
calculating the 
probability that a given triplet $t=ab|c$
is consistent with a random labelling having already specified the labels of leaves $1,2,\ldots,k-1$. Let $\gamma$ be this partial labelling.
We need to consider six cases:

\begin{enumerate}
\item Labels $a,b,c$ have not been assigned yet. The probability is then simply the number of consistent triplets $xy|z$ with $k\leq x,y,z$
divided by $3\binom{n-k+1}{3}$.

\item $\gamma(x)=a$ but both $b$ and $c$ have not been assigned yet. The probability is the number of $k\leq y,z$ such that $xy|z$ is
a consistent triplet divided by $2\binom{n-k+1}{2}$.

\item $\gamma(z)=c$ but both $a$ and $b$ have not been assigned yet. The probability is the number of $k\leq x<y$ such that $xy|z$ is a
consistent triplet divided by $\binom{n-k+1}{2}$.

\item $\gamma(x)=a,\gamma(y)=b$ but $c$ has not been assigned yet. The probability is the number of $k\leq z$ such that $xy|z$ is a consistent
triplet divided by $n-k+1$.

\item $\gamma(x)=a,\gamma(z)=c$ but $b$ has not been assigned yet. The probability is the number of $k\leq y$ such that $xy|z$ is a consistent
triplet divided by $n-k+1$.

\item $\gamma(x)=a,\gamma(y)=b,\gamma(z)=c$. The probability is $0$ or $1$, depending on whether $xy|z$ is consistent or not.

\end{enumerate}

In the first five cases, we can do rather better than simply counting all the $xy|z$ each time from scratch. Define:
\begin{eqnarray*}
count3(k)&:=&\#\text{(x,y,z) such that } k\leq x,y,z\text{ and } x<y\text{ and } consistent(x,y,z)\\
count2x(k,x)&:=&\#\text{(y,z) such that } k\leq y,z\text{ and } consistent(x,y,z)\\
count2z(k,z)&:=&\#\text{(x,y) such that } k\leq x<y\text{ and } consistent(x,y,z)\\
count1xy(k,x,y)&:=&\#\text{z such that } k\leq z\text{ and } consistent(x,y,z)\\
count1xz(k,x,z)&:=&\#\text{y such that } k\leq y\text{ and } consistent(x,y,z)
\end{eqnarray*}
It is easy to see that we can compute all the above values in time $O(m^3)$ as follows:
\begin{eqnarray*}
count3(k)&=&count3(k+1) + count2x(k+1,k) + count2z(k+1,k)\\
count2x(k,x)&=&count2x(k+1,x) + count1xz(k+1,x,k) + count1xy(k+1,x,k)\\
count2z(k,z)&=&count2z(k+1,z) + count1xz(k+1,k,z)\\
count1xy(k,x,y)&=&count1xy(k+1,x,y) + consistent(x,y,k)\\
count1xz(k,x,z)&=&count1xz(k+1,x,z) + consistent(x,k,z)
%count3(k)&=&count3(k+1) + \sum_{k\leq x<y} consistent(x,y,k) + \sum_{k\leq y,z} consistent(k,y,z)\\
%count2x(k,x)&=&count2x(k+1,x) + \sum_{k\leq y} consistent(x,y,k) + \sum_{k\leq z} consistent(x,k,z)\\
%count2z(k,z)&=&count2z(k+1,z) + \sum_{k<y} consistent(k,y,z)\\
%count1xy(k,x,y)&=&count1xy(k+1,x,y) + consistent(x,y,k)\\
%count1xz(k,x,z)&=&count1xz(k+1,x,z) + consistent(x,k,z)
\end{eqnarray*}

Having preprocessed the above values, we can calculate each probability in time $O(1)$, giving us a total running time
of $O(m^3+n^2|T|)$. It turns out, however, that we can do slightly better. As it currently stands, for each leaf and each
unused label we calculate the expected number of consistent triplets after assigning this label to this leaf separately. To further
improve the complexity of the running time we can try to do all such calculations at once (for a fixed leaf): for each triplet and for
each unused label we calculate the probability that this triplet is consistent after we use this label. (In other words, we
switch from a leaf-label-triplet loop nesting to leaf-triplet-label). Then it turns out that those probabilities are the 
same for almost all unused labels. More formally:

\begin{lemma} \label{lem:label_choice}
Let $\Gamma$ be the set of labellings that assign $\gamma(i)$ to leaf $i$ for each $i=1,2,\ldots,k-1$ and $\Gamma_x$ be the subset of
$\Gamma$ that assign $x$ to leaf $k$. We can find $x$ for which $E(f(N,\gamma_x,T))$ is maximum in time $O(|T|)$ assuming the
above preprocessing.
\end{lemma}

\begin{proof}
Let $e_x=E(f(N,\gamma_x,T))$. Each such $e_x$ is a sum of probabilities corresponding to the elements of $T$. We will start with
all $e_x$ equal $0$ and consider triplets one by one. For each triplet $t$ we will increase the appropriate values of $e_x$ by the corresponding
probabilities. Again we should consider six cases depending on how $t=ab|c$ looks like. Let $A=\frac{1}{3{n-k+1 \choose 3}}$,
$B=\frac{1}{{n-k+1 \choose 2}}$ and $C=\frac{1}{n-k+1}$:

\begin{enumerate}

\item Labels $a,b,c$ have not been assigned yet. We should increase $e_a$ and $e_b$ by $count2x(k+1,k)\frac{B}{2}$, $e_c$ by $count2z(k+1,k)B$ and
the remaining $e_l$ by $count3(k+1)A$.

\item $\gamma(x)=a$ but both $b$ and $c$ have not been assigned yet. We should increase $e_b$ by $count1xy(k+1,a,k)C$, $e_c$ by $count1xz(k+1,x,k)C$
and the remaining $e_l$ by $count2x(k+1,x)\frac{B}{2}$.

\item $\gamma(z)=c$  but both $a$ and $b$ have not been assigned yet. We should increase $e_a$ and $e_b$ by $count1xz(k+1,k,z)C$ and the
remaining $e_l$ by $count2z(k,z)B$.

\item $\gamma(x)=a, \gamma(y)=b$ but $c$ has not been assigned yet. We should increase $e_c$ by $consistent(x,y,k)$ and the remaining
$e_l$ by $count1xy(k+1,x,y)C$.

\item $\gamma(x)=a, \gamma(z)=c$ but $b$ has not been assigned yet. We should increase $e_b$ by $consistent(x,k,z)$ and the remaining
$e_l$ by $count1xz(k+1,x,z)C$.

\item $\gamma(x)=a, \gamma(y)=b, \gamma(z)=c$. We should increase all $e_l$ by $consistent(x,y,z)$.

\end{enumerate}

The na\"{i}ve implementation of the above procedure would require time $O(|T|n)$. We can improve it by observing that
we are interested only in the relative increment of the different $e_x$, not in their actual values. So, instead of 
increasing all $e_l$ with $l\notin S$ by some $\delta$ (for some $S \subseteq X$), we can decrease all $e_l$ with $l\in S$ by 
this $\delta$. Then processing each triplet takes time $O(1)$ as we only have to change at most $3$ values of $e_x$.

\qed
\end{proof}

We may now compose these lemmas into the following.

\begin{proof}[of Theorem~\ref{thm:genericderand}]
Consider the procedure sketched at the beginning of the subsection and 
implemented according to the above lemmas. By Lemma~\ref{lem:quality} it produces good labellings.
By Lemma~\ref{lem:qubic_consistency} and $n$ applications of Lemma~\ref{lem:label_choice},
it can be implemented to run in $O(m^3+n|T|)$ time.
\qed
\end{proof}

\subsection{Consequences of Theorem~\ref{thm:genericderand}}
\label{subsec:discuss}

The above theorem gives a new perspective on the problem of approximately constructing phylogenetic networks.
From the algorithm of G\c{a}sieniec et al.~\cite{Gasieniec99} 
we can always construct a phylogenetic tree that is consistent with at least $1/3$ of the the input triplets.
In fact, the trees constructed by 
this algorithm are very specific - they are always caterpillars. (A caterpillar
is a phylogenetic tree such that, after removal of leaves, only a directed path remains.)
Theorem~\ref{thm:genericderand} implies that not only caterpillars, but all possible tree topologies
have the property, that given any set of triplets we may 
find in polynomial time a proper assignment of species into leaves with the
guarantee that the resulting phylogenetic tree is consistent with at least a third of the input triplets.\\
\\ The generality of 
Theorem~\ref{thm:genericderand} makes it meaningful not only for trees, but also for any 
other subclass of phylogenetic 
networks (e.g. for level-$k$ networks). Let us assume that we have focussed our attention on a certain subclass of networks. Consider the 
task of designing an algorithm 
that for a given triplet set constructs a network from the subclass consistent with at
least a certain fraction of the given triplets. A worst-case approach as described in 
this section will never give us a guarantee better than the maximum value 
of $\#N$ ranging over all topologies $N$ in the subclass. Therefore, if we 
intend to obtain networks 
consistent with a big fraction of triplets and if our criteria is to 
maximize this fraction in the worst case, then our task reduces to finding 
topologies within 
the subclass that are good for the full triplet set.
Theorem~\ref{thm:genericderand} potentially has a further use as a mechanism for comparing the quality of 
phylogenetic networks generated by other methods, because it provides lower bounds for the fraction of $T$ that a
given topology and/or subclass of topologies
can be consistent with. (Although a fundamental problem in phylogenetics
\cite{framework2004}
\cite{tripartite2007}
\cite{intraspec2005}
\cite{nakleh2003}
\cite{plants2004}
the science of network comparison is still very much in its infancy.)\\
\\
For level-$0$ networks (i.e. phylogenetic trees) the problem of finding optimal
topologies for the full triplet set is simple: any tree is consistent with exactly $1/3$ of the full triplet 
set. For level-$1$ phylogenetic networks a topology that is optimal
for the full triplet set was constructed in~\cite{JS2}.
We may use this network and Theorem~\ref{thm:genericderand}
to obtain an algorithm that works for any triplet set and creates a network that is
consistent with the biggest possible fraction of triplets in the worst case (see Section~\ref{sec:lev1} for
more details). For level-$2$ networks we do not yet know the optimal structure of a topology for the full triplet
set, but we will show in Section~\ref{sec:lev2} that we can construct a network
that has a guarantee of being consistent with at least a fraction $0.61$ of the input triplets.

\section{Application to level-1 phylogenetic networks}
\label{sec:lev1}

\subsection{A worst-case optimal polynomial-time algorithm for level-1 networks}

In \cite{JS2} it was shown how to construct a special level-1 topology $C(n)$,
which we call a \emph{galled caterpillar}\footnote{In \cite{JS2} this is called a
\emph{caterpillar network}.}, such that $\#C(n) \geq \#N$ for
all level-1 topologies $N$ on $n$ leaves. The existence of $C(n)$, which
has a highly regular structure, was proven by showing that any other topology $N$ can be transformed into $C(n)$
by local rearrangements that never decrease the number of triplets the associated network is consistent with. It was shown that
$\#C(n) = S(n)/3\binom{n}{3}$, where $S(0)=S(1)=S(2)=0$ and, for $n > 2$,

\begin{equation}
\label{eq:sn}
S(n) = max_{1 \leq a \leq n} \bigg \{ \binom{a}{3} + 2 \binom{a}{2} (n-a) + a \binom{n-a}{2} + S(n-a) \bigg \}.
\end{equation}

\begin{figure}
\centering \vspace{-0.5cm} \includegraphics[scale=0.4]{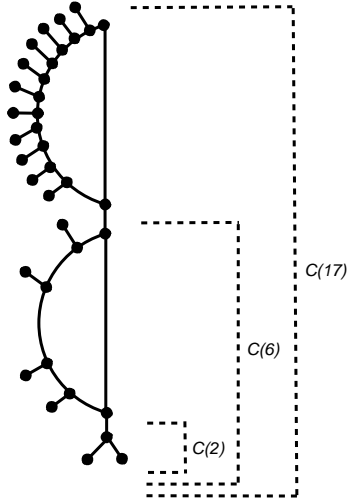} 
\caption{This is galled caterpillar $C(17)$. It contains two galls and 
ends with a tail of two leaves. $C(17)$ contains 11 leaves in the top gall
because Equation \ref{eq:sn} is maximised for $a=11$.}  \vspace{-0.3cm} \label{fig:cats}
\end{figure}

In Figure \ref{fig:cats} an example of a galled caterpillar is shown. All galled caterpillars on $n \geq 3$ leaves
consist of one or more \emph{galls} chained together in linear fashion and terminating in a tail of one or two leaves.
Observe that the recursive structure of $C(n)$ mirrors directly the recursive definition of $S(n)$ in the
sense that the value of $a$ chosen at recursion level $k$ 
is equal to the number of leaves found in the $k$th gall, counting downwards from the root. In the
definition of $C(n)$ it is not specified how the
$a$ leaves at a given recursion level are distributed within the gall, but it is easy to verify that placing
them all on one side of the gall (as shown in the figure) is sufficient.

\begin{lemma}
\label{lem:lev1}
Let $T$ be a set of input triplets labelled by $n$ species. Then in time $O(n^3 + n|T|)$ it is possible to construct
a level-1 network $\mathcal{N}$, isomorphic to the galled caterpillar $C(n)$, consistent with at least a fraction
$S(n)/3\binom{n}{3}$ of $T$. 
\end{lemma}
\begin{proof}
First we construct the level-1 topology $C(n)$. Using dynamic programming to compute all values of $S(n')$ for $0 \leq n' \leq n$
we can do this in time $O(n^2)$. Note that $C(n)$ contains in total $O(n)$ vertices. It remains only to choose an
appropriate labelling of the leaves of $C(n)$, and this is achieved by substituting $C(n)$ for $N$ in Theorem 
\ref{thm:genericderand}; this dominates the running time. \qed
\end{proof}

Note that, because $C(n)$ achieves the best possible fraction for the input $T_1(n)$, the fraction achieved
by Lemma \ref{lem:lev1} is worst-case optimal for all $n$. Empirical experiments suggest that the function $S(n)/3\binom{n}{3}$
is strictly decreasing and approaches a horizontal asymptope of $0.4880...$ from above; for values of $n=10^1, 10^2, 10^3, 10^4$
the respective ratios are $0.511..., 0.490..., 0.4882..., 0.4880...$. It is difficult to formally prove convergence
to 0.4880... so we prove a slightly weaker lower bound of 0.48 on this function. From this it follows that in all cases the 
algorithm described in Lemma \ref{lem:lev1} is guaranteed to produce a network consistent with at least a fraction $0.48$ of $T$, improving 
considerably on the  $5/12 \approx 0.4166$ fraction achieved in \cite{JS2}.

\begin{lemma}
\label{lem:asymp}
$S(n)/3\binom{n}{3} > 0.48$ for all $n \geq 0$.
\end{lemma}
\begin{proof}
This can easily be computationally verified for $n < 116$, we have done this with a computer program
written in Java \cite{myurl}. To prove it for $n \geq 116$, assume by 
induction that
the claim is true for all $n' < n$. Instead of choosing the value of $a$ that maximises $S(n)$ we
claim that setting $a$ equal to $z=\lfloor 2n/3 \rfloor$ is sufficient for our purposes. We thus need to prove
the following inequality:
\[
\frac{ \binom{ z }{3} + 2\binom{z}{2}(n-z) + z\binom{n-z}{2} + S(n-z) }{ 3\binom{n}{3} } > 48/100.
\]
Combined with the fact that, by induction, $S(n-z)/3\binom{n-z}{3} > 48/100$, it is sufficient to prove that:
\[
\frac{ \binom{ z }{3} + 2\binom{z}{2}(n-z) + z\binom{n-z}{2} + 144/100\binom{n-z}{3} }{ 3\binom{n}{3} } > 48/100
\]
Using Mathematica we rearrange the previous inequality to:
\[
\frac{ \lfloor 2n/3 \rfloor \bigg ( 22 + 9n + 33n^2 - 6(7+18n)\lfloor 2n/3 \rfloor + 86 \lfloor 2n/3 \rfloor^2 \bigg ) }{n(2-3n+n^2)} < 0
\]
Taking $(2n/3)-1$ as a lower bound on $\lfloor 2n/3 \rfloor$, and $2n/3$ as an upper bound, it can be easily verified that
the above inequality is satisfied for $n \geq 116$.
%The denominator is positive for $n \geq 3$ so we only need to prove:
%\begin{align*}
%22 + 9n + 33n^2 - 6(7+18n)\lfloor 2n/3 \rfloor + 86 \lfloor 2n/3 \rfloor^2 & < 0 \\
%\]
%Given that $(2n/3 - 1) \leq \lfloor 2n/3 \rfloor \leq 2n/3$, proving the following is sufficient:
%\[
%\Leftarrow 22 + 9n + 33n^2 - 6(7+18n)(2n/3 - 1) + 86(2n/3)^2  & < 0 \\
%\Leftarrow n & \geq 116.
%\end{align*}
\qed
\end{proof}

\noindent
To conclude this section we combine Lemmas \ref{lem:lev1} and \ref{lem:asymp} into the following Theorem.

\begin{theorem}
\label{thm:lev1summary}
Let $T$ be a set of input triplets labelled by $n$ species.
In time $O(n^3 + n|T|)$ it is possible to construct a level-1 network $\mathcal{N}$ consistent with at least
a fraction $S(n)/3\binom{n}{3} > 0.48$ of $T$, and this is worst-case optimal.
\end{theorem}

\section{A lower bound for level-2 networks}
\label{sec:lev2}

\begin{theorem}
\label{thm:lev2}
Let $T$ be a set of input triplets labelled by $n$ species.
It is possible to find in polynomial time a level-2 network $\mathcal{N}(T)$ such
that $\mathcal{N}(T)$ is consistent with at least a fraction 0.61 of $T$.
\end{theorem}

\begin{proof}
%We prove this by showing how to construct a network $LB_2(n)$ consistent
%with at least 0.61 of the triplets in $T_1(n)$. Using Theorem \ref{thm:genericderand}
%this can then be labelled to obtain the network $N(T)$.
We prove this by induction.
For $n < 16813$ we use a computational proof. For $n \geq 16813$ we use Mathematica
to show that, assuming the induction base, a fraction 0.61 can be achieved by repeatedly chaining together
a very basic type of level-2 network into some kind of ``level-2 caterpillar''; the
details are deferred to the appendix. \qed
\end{proof}

%                long midC = (long) Math.floor(0.385 * n);
%                long midD = (long) Math.floor(0.07 * n);
%                long midE = (long) Math.floor(0.26 * n);
%                long midF = (long) Math.floor(0.285 * n);

%\section{Efficient derandomization}
%\label{sec:efficient}
%
%PAWEL'S TEX HERE

\section{The complexity of optimisation}
\label{sec:complexity}

Given a topology $N$ and a set of triplets $T$, the techniques described in this article guarantee to find a 
labelling 
$\gamma$ such that $f(N,\gamma,T) \geq \#N$. It is natural to explore the complexity of finding, in polynomial time, (approximations 
to) an optimal labelling of $N$ for a particular triplet set $T$. The observation below rules out (under standard 
complexity-theoretic assumptions) the existence of a \emph{Polynomial-Time Approximation Scheme} (PTAS) for this problem. Secondly we 
observe that a PTAS for the problem MAX-LEVEL-0 (which carries the name MCTT in \cite{wu04}) can also be ruled out. We discuss the 
consequences of this in the next section.\\
\\
\textbf{Problem: } MAX-LEVEL-0-LABELLING\\
\textbf{Input: } A level-0 topology $N$ (i.e. a topology of a phylogenetic tree) and a set $T$ of rooted triplets.\\
\textbf{Output: } The maximum value of $s$ such that there exists a labelling $\gamma$ of $N$ making
$(N, \gamma)$ consistent with at least $s$ triplets from $T$.\\
\\
\textbf{Problem: } MAX-LEVEL-0\\
\textbf{Input: } A set $T$ of rooted triplets.\\
\textbf{Output: } The maximum value of $s$ such that there exists a level-0 network
$\mathcal{N}$ (i.e. a phylogenetic tree) consistent with at least $s$ triplets from $T$.

\begin{observation}
\label{obs:apxhard}
MAX-LEVEL-0 and MAX-LEVEL-0-LABELLING are both APX-complete.
\end{observation}
\begin{proof}
By approximation-preserving reduction from MAXIMUM SUBDAG / LINEAR ORDERING. See appendix. \qed
\end{proof}

\section{Conclusions and open questions}
\label{sec:conc}

With Theorem \ref{thm:genericderand} we have described a method which shows how, in polynomial time, good solutions
for the full triplet set can be efficiently converted into equally good, or better, solutions for more general
triplet sets. Where best-possible solutions are known for the full triplet set, this leads
to worst-case optimal algorithms, as demonstrated by Theorem \ref{thm:lev1summary}. An
obvious next step is to use this method to generate
algorithms (where possible worst-case optimal) for wider subclasses of phylogenetic
networks. Finding the/an ``optimal form'' of level-2 networks for the full triplet set remains
a fascinating open problem.

From a biological perspective (and from the
perspective of understanding the relevance of triplet methods) it is also important
to attach \emph{meaning} to the networks that the techniques described in this paper produce. For example, we have
shown how, for level-1 networks, we can always find a network isomorphic to a galled caterpillar which is consistent with
at least a fraction 0.48 of the input. If we do this, does the location of the species within this galled caterpillar
communicate any biological information? Also, what does it say about the relevance of triplet methods, and especially
the level-$k$ hierarchy, if we know \emph{a priori} that a large fraction (already for level 2 more than 0.61) 
of the input can be made consistent with some network from the subclass? And, as discussed in Section \ref{subsec:discuss},
how far can the techniques described in this paper be used as a quality measure for networks produced by other
algorithms?

As mentioned in the introduction, an algorithm guaranteed to find a network consistent with a fraction $p'$ of the input trivially 
becomes a $p'$-approximation for the MAX variant of the problem (where we optimise not with respect to $|T|$ but with respect to the 
size of the optimal solution for $T$.) In fact, the best-known approximation factor for MAX-LEVEL-0 is 1/3, a trivial extension of 
the fact that $p=1/3$ for trees \cite{Gasieniec99}. On the other hand, the APX-hardness of this problem implies that an 
approximation factor arbitrarily close to 1 will not be possible. It remains a highly challenging open problem to determine 
whether better approximation factors can be obtained for the latter problem via some different approach. (We note that an 
approximation
factor $f > 1/2$ would be a major breakthrough because, by the proof of Observation \ref{obs:apxhard}, this would improve upon
the long-standing best-known approximation factor of $1/2$ for the problem LINEAR ORDERING \cite{fences}.) Alternatively, there could be some complexity theoretic 
reason why approximation factors better than $p$ (where $p$ is optimal in our formulation) are  not possible. Under strong complexity-theoretic assumptions the best 
approximation factor possible for MAX-3-SAT, for example, uses 
a trivial upper bound of all the clauses in the input \cite{hastad}, analogous perhaps to using $|T|$ as an upper bound.

\section{Acknowledgements}

We thank Leo van Iersel and Jesper Jansson for very helpful comments.

\clearpage

\appendix

\section{Appendix}

In this section we give the full proofs for Lemma \ref{lem:fast_consistency} from Section \ref{sec:label},
Theorem \ref{thm:lev2} from Section \ref{sec:lev2} and
Observation \ref{obs:apxhard} from Section \ref{sec:complexity}.\\

\begin{figure}
\centering \vspace{-0.5cm} \includegraphics{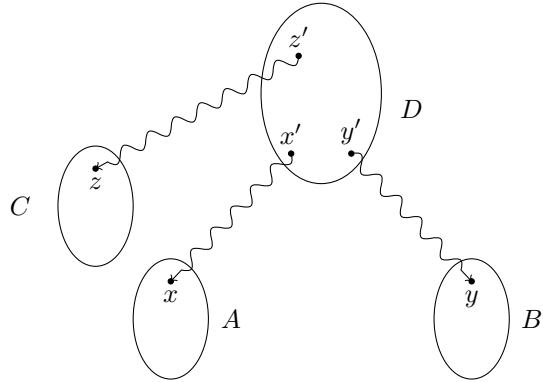} \caption{The case $LCA(A,B)=LCA(A,B,C)$ from the proof of Lemma \ref{lem:fast_consistency}.} \vspace{-0.3cm} 
\label{fig:lca}
\end{figure}

\noindent
\textbf{Lemma \ref{lem:fast_consistency}.}
\emph{Given a level-$k$ network $\mathcal{N}$, we can preprocess it in time $O(m + mk^2)$ so that given any three vertices $x,y,z$ we can check 
whether $xy|z$ is consistent with $\mathcal{N}$ in time $O(1)$.}

\begin{proof}
We begin with a definition and two claims. A \emph{chain} is a sequence of vertices $v_0\rightarrow v_1\rightarrow\ldots v_{k+1}$ such that both in and out degrees 
of $v_1,\ldots,v_k$ are all equal 1. Now, consider an arbitrary biconnected component in $\mathcal{N}$. We claim that, if each chain within the component
is replaced by an edge, the transformed component contains at most $\max(2,3k)$ vertices. This is trivially true for the biconnected component consisting
of only a single edge. To see that it is more generally true, note that the transformed component contains at most $k$ recombination vertices and
that all other vertices have outdegree 2 and indegree at most 1. Each such non-recombination vertex creates at least one new path that eventually
has to terminate in a recombination vertex, and a recombination vertex can terminate at most two paths. So there are at most $2k$ non-recombination
vertices, and the claim follows.

From now on biconnected component (or simply component) refers to a biconnected component containing more than
one edge. We claim that all biconnected components within $\mathcal{N}$ are vertex-disjoint. To see why this is, note that each vertex in such a 
component must have degree at least 2. The value of indegree plus outdegree of all vertices in $\mathcal{N}$ is at most 3, so they cannot be a part of more than
one such component.

The above reasoning shows that we could imagine the network as a rooted tree $\mathcal{T}$ in which each vertex corresponds to some bigger
component that has relatively simple structure: contracting all chains in it gives us a graph of size $O(k)$.
%In each such component we
%distinguish its ``highest'' vertex, which will either be the root of the whole network or the end of the edge connecting the component to its parent in this 
%rooted tree. We will call it the {\it representative} of the component.

First we focus on the case when $x,y,z$ lie in one biconnected component. The obvious solution is to simply preprocess such queries for each
triple of vertices from one biconnected component. This does not give us the desired complexity yet: while each component is small after contracting chains, it 
might originally contain as many as $\Omega(m)$ vertices. We may overcome this difficulty by observing that if some chain consist of more than $3$ inner vertices, 
we can replace it by a chain containing exactly $3$ of them. Then we preprocess the resulting graph using the $O(|V|^3)$ algorithm
from Lemma \ref{lem:qubic_consistency} (observe that $|V|=O(k)$). Given a query concerning 
consistency of some $xy|z$ we may have to replace some of $x,y,z$ with other vertices lying on the shortened chains, which can be easily done in time $O(1)$. The 
whole preprocessing takes time $\sum_i O(k_i^3)$ where all $k_i=O(k)$ and $\sum_i k_i=O(m)$, giving us the desired complexity.

Now we can solve the general case. Let $A,B,C$ be biconnected components such that $x\in A$, $y\in B$ and $z\in C$. We can preprocess $\mathcal{T}$
in linear time so that given its two vertices we can find their lowest common ancestor (\emph{LCA}) in time $O(1)$ \cite{bender1}. (This
is already sufficient for the case $k=0$ and contributes the $m$ term in the running time.) If $xy|z$ is 
consistent there are 
only two possible situations:
\begin{enumerate}
\item $LCA(A,B,C)$ is a proper ancestor of $LCA(A,B)$. This is easy to detect.
\item $LCA(A,B)=LCA(A,B,C)$. Let $D=LCA(A,B)$. We can find $x',y',z'$ - the entrance points
within $D$ (see Figure \ref{fig:lca}) - in time $O(1)$. It can be done by either using level ancestor queries or 
modifying the \emph{LCA}
algorithm \cite{bender2}. We can then check whether $x'y'|z'$ is consistent using the above preprocessing.
\end{enumerate}
\qed
\end{proof}

\begin{figure}
\centering \vspace{-0.5cm} \includegraphics[scale=0.5]{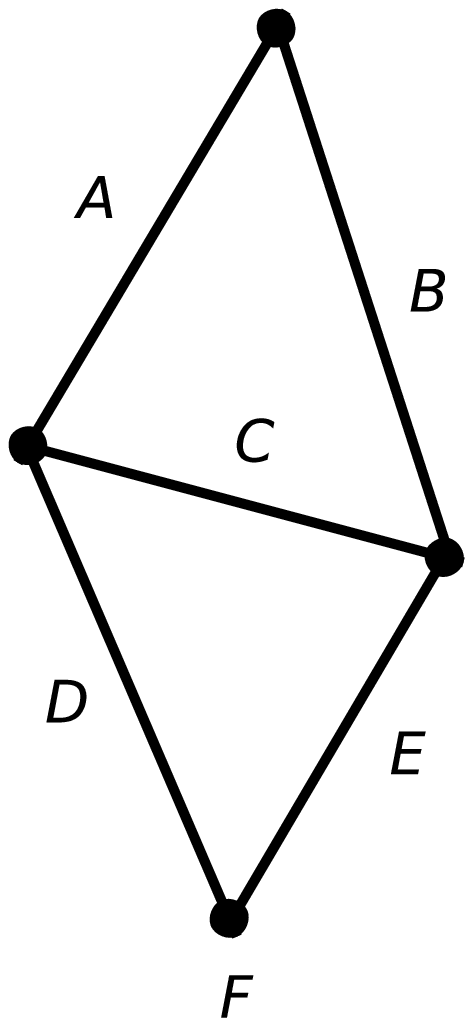} 
\hspace{3cm} \includegraphics[scale=0.9]{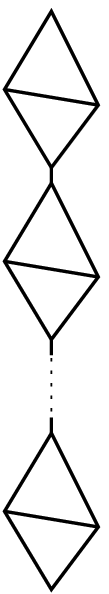}
\vspace{+0.1cm} 
\caption{We construct the network $LB_2(n)$ by repeatedly chaining the structure on the
left, a \emph{simple level-2 network} (see \cite{iersel}), together to obtain an overall topology resembling the structure on the 
right.}
\vspace{-0.1cm} \label{fig:8a}
\end{figure}

\noindent
\textbf{Theorem \ref{thm:lev2}.} \emph{Let $T$ be a set of input triplets labelled by $n$ species.
It is possible to find in polynomial time a level-2 network $\mathcal{N}(T)$ such
that $\mathcal{N}(T)$ is consistent with at least a fraction 0.61 of $T$.}\\[-6pt]
\begin{proof}
We prove the theorem by showing how to construct a topology, which we call $LB_2(n)$, consistent
with at least 0.61 of the triplets in $T_1(n)$. Using Theorem \ref{thm:genericderand}
$LB_2(n)$ can then be labelled to obtain the network $\mathcal{N}(T)$. We show by induction
how $LB_2(n)$ can be constructed. We take $n < 16813$ as the induction base; for these values of $n$ we refer
to a simple computational proof written in Java \cite{myurl}. We now prove the result for $n \geq 16813$.
Let us assume by induction that, for any $n' < n$, there exists some topology 
$LB_2(n')$ such that $\#LB_2(n') \geq 0.61$. If we let $t(n')$ equal the
number of triplets in $T_1(n')$ consistent with $LB_2(n')$, we have that $t(n')/3\binom{n'}{3} \geq
0.61$ and thus that $t(n') \geq 1.83 \binom{n'}{3}$. Consider the structure in Figure \ref{fig:8a}. For $S \in \{A,B,C,D,E\}$,
we define the operation \emph{hanging $l$ leaves from side $S$} as replacing the edge $S$ with a directed
path containing $l$ internal vertices, and then attaching a leaf to each internal vertex.
We construct $LB_2(n)$ as follows. We create a copy of the structure from the figure and
hang $c = \lfloor 0.385n \rfloor$ leaves from side $C$, $d=\lfloor 0.07n \rfloor$ from
side $D$ and $e = \lfloor 0.26n \rfloor$ from side $E$. We let $f = \lfloor 0.285n \rfloor$ and
add the edge $(F,r)$, where $r$ is the root of the network $LB_2(f)$. Finally we
hang $a = n-(c+d+e+f)$ leaves from side $A$; it might be that $a=0$. (The only reason we hang
leaves from side $A$ is to compensate for the possibility that $c+d+e+f$ does not
exactly equal $n$.) This completes the construction of $LB_2(n)$; note that as in Section \ref{sec:lev1}
the network is constructed by recursively chaining the same basic structure together.

We can use Mathematica
to show that $LB_2(n)$ is consistent with at least $0.61$ of the triplets in $T_1(n)$.
In particular, by explicitly counting the triplets consistent with $LB_2(n)$ we derive an 
inequality expressed in terms of $n, c, d, e, f, t(f)$, which Mathematica then simplifies to a cubic 
inequality in $n$ that holds for all $n \geq 16813$. (To simplify
the inequality we take $x-1$ as a lower bound on $\lfloor x \rfloor$ and assume that no leaves are
hung from side A). The Mathematica script is reproduced in Figure \ref{fig:math}, and can be
downloaded from \cite{myurl}. Finally, we comment that the networks computationally constructed
for $n < 16813$ are, essentially, built
in the same way as the networks described above. The only difference is that, to absorb inaccuracies
arising from the floor function, we try several possibilities for how many leaves should be hung
from each side; for side $C$, for example, we try also $(c-1)$ and $(c+1)$ leaves. \qed
\end{proof}

\begin{figure}
\centering \vspace{-0.5cm} \includegraphics{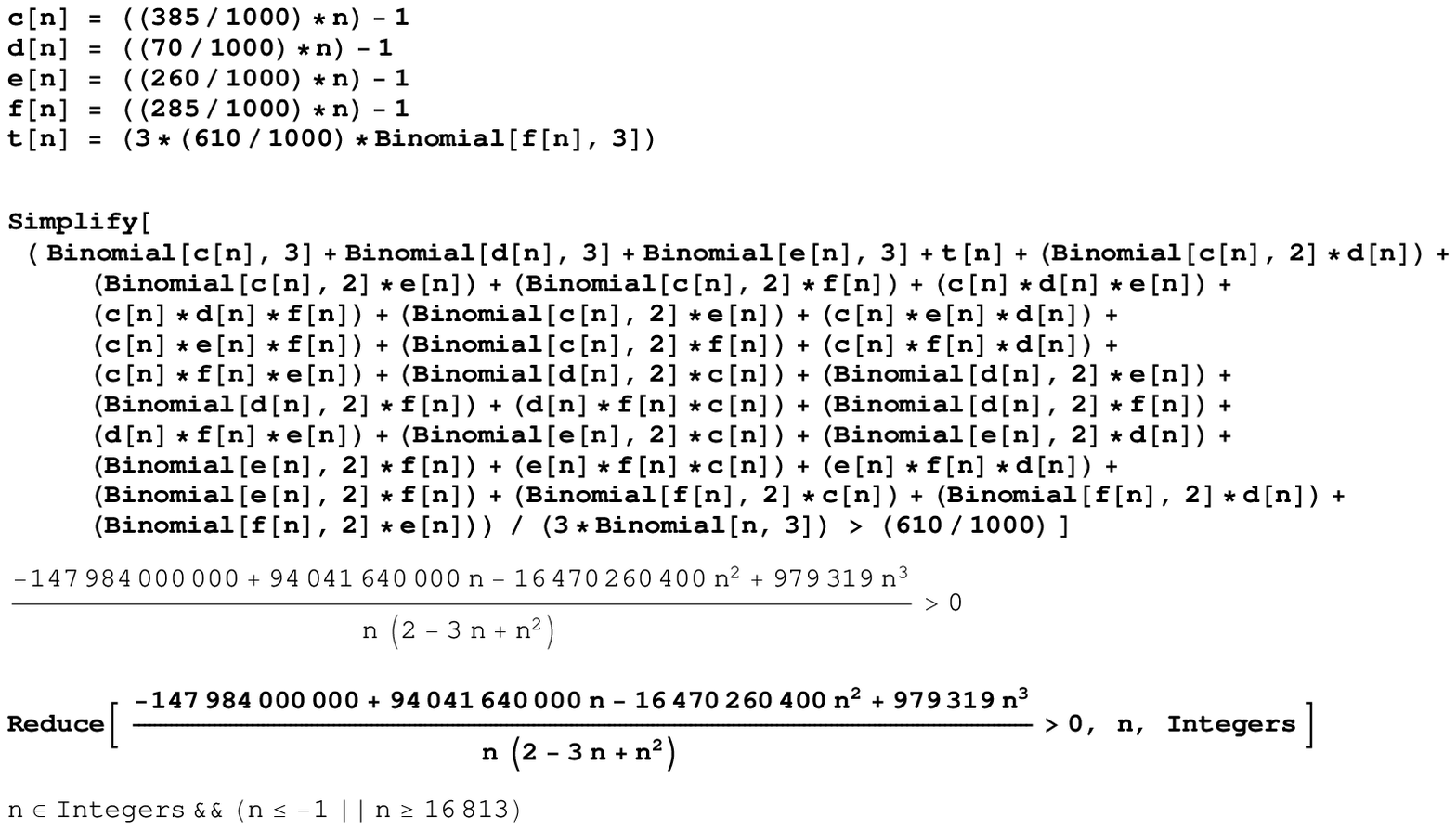}
\caption{The Mathematica script used in the proof of Theorem \ref{thm:lev2}.}
\vspace{-0.3cm} \label{fig:math}
\end{figure}

\noindent
\textbf{Observation \ref{obs:apxhard}. } \emph{MAX-LEVEL-0 and MAX-LEVEL-0-LABELLING are both APX-complete.}\\[-6pt]
\begin{proof}
By Theorem 1 we may label any tree topology to make
it consistent with 1/3 of the given triplets. Therefore,
both problems are in the class APX.
To prove APX-hardness we use a reduction proposed by Wu~\cite{wu04}
and we show that it is actually an \emph{L-reduction}
from the MAXIMUM SUBDAG problem. Both L-reductions and
the MAXIMUM SUBDAG problem were studied by Papadimitiou 
and Yannakakis~\cite{papa}. MAXIMUM SUBDAG, which is equivalent
to the problem LINEAR ORDERING, has been proven
APX-complete \cite{fences}\cite{papa}.

In the MAXIMUM SUBDAG problem we are given a directed graph 
$G = (V,A)$, and the goal is to find a maximal cardinality
subset of arcs $A' \subset A$ such that $G' = (V,A')$ is acyclic.

In the reduction of Wu one constructs an instance of the 
MAX-LEVEL-0 problem as follows. Given a directed graph $G = (V,A)$,
let $x \notin V$, consider the set of triplets $T$
containing a single triplet $t_{uv} = ux|v$ for every arc $(u,v) \in A$,
where $X = X(T) = V \cup \{ x\}$.
To argue that it is an L-reduction it remains to
prove the following two claims.

1) If there exists a subset of arcs $A' \subset A$ such that $G' = (V,A')$ is acyclic
and $|A'| = k$, then there exists a phylogenetic tree consistent with at least 
$k$ triplets from $T$. To prove this claim, we consider a topological 
sorting of vertices in graph $G'$. We construct the phylogenetic tree 
to be a caterpillar with the leaves labeled (bottom-up) by such sorted vertices,
the lowest leaf is labelled by $x$.
It remains to observe that for any arc $(u,v) \in A'$ the corresponding
triplet $t_{uv}$ is consistent with the obtained phylogenetic tree.   

2) Given a phylogenetic tree $\mathcal{B}$ consistent with $l$ triplets from $T$,
we may construct in polynomial time a subset of arcs $A' \subset A$ 
such that $G' = (V,A')$ is acyclic and $|A'| = l$. In fact we will show
that it suffices to take $A'$ consisting of the arcs $(u,v)$ such that
the corresponding triplet $t_{uv}$ is consistent with $\mathcal{B}$.
We only need to argue that for such a choice of $A'$ the resulting
graph $G'=(V,A')$ is acyclic. Consider the path in the tree $\mathcal{B}$ from the root 
node to the leaf labeled by the special species $x$.
For any vertex $v \in A$, the species $v$ has an internal node 
on this path where he branched out of the evolution of $x$, namely
the lowest common ancestor of $u$ and $x$ ($LCA(u,x)$).
Observe, that the position of $LCA(u,x)$ induces a partial ordering $>_{\mathcal{B}}$ on $A$.
Recall, that if a triplet $t_{uv} = ux|v$ is consistent with $\mathcal{B}$, then
$LCA(u,x)$ is a proper ancestor of $LCA(v,x)$. Therefore, the consistent 
triplets from $T$ induce another partial ordering that may be 
extended to $>_{\mathcal{B}}$. This implies that for $A'$ 
containing the arcs $(u,v)$ such that
a triplet $t_{uv}$ is consistent with $\mathcal{B}$
the graph $G' = (V,A')$ is acyclic.

With the above construction we have shown that the existence
of an $\epsilon$-approximation algorithm for the MAX-LEVEL-0
problem implies existence of an $\epsilon$-approximation algorithm
for the MAXIMUM SUBDAG problem. In particular, existence of
a Polynomial-Time Approximation Scheme (PTAS) for MAX-LEVEL-0
would imply existence of PTAS for MAXIMUM SUBDAG, which is unlikely
due to the results in \cite{fences} and \cite{papa}.

In the reduction we might have assumed a particular topology for the tree.
Namely, we might have assumed, that the topology needs to be
a caterpillar. Therefore, the problem MAX-LEVEL-0-LABELLING is also APX-hard. 
\qed
\end{proof}

%\bibliography{flp}

\end{document}